\documentclass{jfpc}
\usepackage{amssymb,amsmath,amsthm}
\usepackage{graphicx,xspace,hyperref,centernot}
\usepackage[T1]{fontenc}
\usepackage[frenchb]{babel}

\usepackage{tikz}
\usetikzlibrary{arrows,calc,petri,positioning,shapes}
\tikzset{
     every place/.style={thick, minimum size=8mm},
     every transition/.style={thick, minimum size=6mm},
     pre/.style={<-,shorten <=1pt,>=stealth',thick},
     post/.style={->,shorten >=1pt,>=stealth',thick},
     round/.style={rounded corners=5pt},
     fire/.style={transition,fill=yellow}
}

\usepackage{listings}

\makeatletter
\newcommand{\CTL}{\@ifnextchar*{\textup{CTL$^*$}}{\textup{CTL}}}
\makeatother

\newcommand{\N}{\mathbf N}

\newtheorem{definition}{Définition}
\newtheorem{proposition}{Proposition}  
  
\newtheorem{example}{Exemple}
\renewenvironment{proof}{\noindent\emph{Preuve.}}{\hfill\qed\bigskip}

\annee{2012}
\titre{Un modèle booléen pour l'énumération des siphons et des pièges minimaux dans les réseaux de Petri}
\auteurs{Faten Nabli \and Fran\c{c}ois Fages \and Thierry Martinez \and Sylvain Soliman} 

\institutions{
EPI CONTRAINTES \\INRIA Paris-Rocquencourt \\Domaine de
Voluceau, Rocquencourt, BP 105,\\ 78153 LE CHESNAY CEDEX - FRANCE}

\mels{\{Faten.Nabli, Francois.Fages, Thierry.Martinez, Sylvain.Soliman\}@inria.fr}

\begin{document}
\creationEntete
\begin{resume}

  Les réseaux de Petri sont un formalisme simple 
  pour modéliser les calculs concurrents.
  Récemment, ils se sont révélés être un outil puissant
  de modélisation et d'analyse des réseaux de réactions 
  biochimiques, en comblant le fossé entre les
  modèles purement qualitatifs et quantitatifs. Les réseaux biologiques
  peuvent être larges et complexes, ce qui rend
  leur analyse difficile.  Dans cet article, nous nous concentrons sur
  deux propriétés structurelles des réseaux de Petri: les
  siphons et les pièges, qui nous apportent des informations sur la
  persistance de certaines espèces biochimiques.  Nous présentons
  deux méthodes pour énumérer les siphons et les pièges minimaux
  d'un réseau de Petri en itérant la résolution d'un problème
  booléen, interprété comme un programme SAT ou PLC(B). Nous comparons les
  performances de ces méthodes avec un algorithme dédié qui
  représente l'état de l'art dans la communauté des réseaux de
  Petri.
  Nous montrons que les programmes SAT et PLC(B) sont plus efficaces.
Nous analysons pourquoi ces programmes sont si performants 
sur les modèles de l'entrepôt de modèles \href{http://www.biomodels.net/}{biomodels.net}
et nous proposons des instances difficiles pour le problème de l'énumération des siphons minimaux.

\end{resume}

\begin{abstract}
Petri-nets are a simple formalism for modeling concurrent computation. 
Recently, they have emerged as a powerful tool for the modeling and analysis of
biochemical reaction networks,
bridging the gap between purely qualitative and quantitative models.
These networks can be large
and complex, which makes their study difficult and computationally challenging.
In this paper, we focus on two
structural properties of Petri-nets, siphons and traps, that
bring us information about the persistence of some molecular species. 
We present two methods for enumerating all minimal siphons
and traps of a Petri-net by iterating the resolution of a boolean
model interpreted as either a SAT or a CLP(B) program.
We compare the performance of these methods with a state-of-the-art
dedicated algorithm of the Petri-net community. 
We show that the SAT and CLP(B) programs are both faster. 
We analyze why these programs perform so well on the models of the
repository of biological models biomodels.net,
and propose some hard instances for the problem of minimal siphons enumeration.

\end{abstract}

\section{Introduction}

Les réseaux de Petri ont été introduits dans les années 60 comme un formalisme simple
pour décrire et analyser les systèmes concurrents, asynchrones, non déterministes et éventuellement distribués.

L'utilisation des réseaux de Petri pour représenter les modèles de réactions biochimiques,
en formalisant les espèces moléculaires par les places et les réactions par les transitions,
a été introduite assez tardivement dans\, \cite{RML93ismb} ,
en proposant certains concepts et outils des réseaux de Petri pour l'analyse de ces
réseaux biochimiques.
Dans\, \cite{Soliman08wcb}, un programme logique avec contraintes sur domaines finis (PLC(DF))
est proposé pour le calcul des P-invariants.
Ces invariants fournissent des lois de conservation structurelles qui peuvent être utilisées
pour réduire la dimension des équations différentielles ordinaires (EDO) associées à un modèle de réactions
avec expressions cinétiques.

Dans ce papier, nous considérons les concepts de siphons et 
pièges des réseaux de Petri.
Ces structures sont liées au comportement dynamique du réseau: la vérification de l'accessibilité (est-ce que le système peut atteindre un état
donné) et à la vivacité (absence de blocages).
Ces propriétés ont été déjà utilisées pour l'analyse des
réseaux métaboliques \cite{ZS03insilicobio}.
Un programme linéaire en nombres entiers est proposé dans\, \cite{CFP02ieee} 
et un algorithme dédié est décrit dans\,\cite{CFP05ieee} pour les calculer.

Un siphon est un ensemble de places qui, une fois non marqué, le reste.
Un piège est un ensemble de places qui, une fois marqué, ne peut jamais perdre tous ses jetons.
Un exemple typique de siphon est un ensemble de métabolites qui sont progressivement réduits
au cours d'une famine;
un exemple typique de piège est l'accumulation de métabolites
qui sont produits au cours de la croissance d'un organisme.

Dans cet article, après quelques préliminaires sur les réseaux de Petri, les siphons et les
pièges, nous donnons un modèle booléen simple de ces propriétés.
Nous décrivons deux méthodes pour énumérer l'ensemble des siphons et des pièges minimaux
et nous les comparons avec un algorithme dédié qui représente l'état de l'art 
\cite{CFP05ieee},
sur un banc d'essai composé à la fois du dépôt  \href{http://www.petriweb.org/}{Petriweb}
\cite{GHPW06icatpn} et de l'entrepôt de modèles \href{http://www.biomodels.net/}{biomodels.net}
\cite{NBBCDDLSSSSH06nar}.
Dans la dernière section, 
nous expliquons pourquoi les programmes proposés sont très performants sur les modèles biochimiques de l'entrepôt Biomodels.net et proposons certaines instances difficiles pour le problème de l'énumération des siphons minimaux.
\section{Préliminaires}
\label{sec:prelim}
\subsection{Réseaux de Petri}
\label{sec:pn}

Un réseau de Petri $PN$ est un graphe biparti orienté $ PN = (P, T, W) $,
où $ P $ est un ensemble fini de sommets appelés places, $ T $ est un ensemble fini de sommets (disjoint de $ P $) appelés transitions
et $ W: ((P \times T) \cup (T \times P)) \to \N $
représente un ensemble d'arcs orientés et pondérés par des entiers
(le poids zéro représente l'absence d'arc).
Un marquage d'un graphe de réseau de Petri est une fonction $m: P \to \N$, qui affecte
un certain nombre de jetons à chaque place.
Un réseau de Petri marqué est un 4-tuple $ (P, T, W, m_0) $ où $ (P, T, W) $ est un réseau de Petri et 
$m_0 $ est un marquage initial. 

\newcommand*{\Pred}[1]{{^\bullet}#1}
\newcommand*{\Succ}[1]{#1^\bullet}

L'ensemble des prédécesseurs (resp.~successeurs) d'une transition $ t \in T $ est l'ensemble
des places $ \Pred t = \{p \in P \mid W(p, t)> 0 \} $
(resp.~$\Succ t =\{p \in P \mid W(t,p)>0\}$).
De même, l'ensemble des prédécesseurs (resp.~successeurs) d'une place $ p \in P $ est l'ensemble de transitions
$\Pred p =\{t \in T \mid W(t,p)>0\}$
(resp.~$\Succ p =\{t \in T \mid W(p,t)>0\}$).

Pour tous marquages $m, m': P \to \N$ et toute transition $t\in T$, 
il y a une étape de transition $m \overset t \rightarrow m'$,
si et seulement si pour tout $ p \in P $,
$ m(p) \geq W (p, t) $ et $ m'(p) = m(p)-W (p, t) + W (t, p) $.

Cette notation reste valable pour une séquence de transitions $\sigma=(t_0 \dots t_n)$
en écrivant
$m \overset \sigma \rightarrow m'$
si 
$m \overset {t_0} \rightarrow m_1
 \overset {t_1} \rightarrow \dots
 \overset {t_{n-1}} \rightarrow  m_{n}
\overset {t_n} \rightarrow m'$
pour les marquages $m_1, \dots, m_{n}$.

La représentation classique d'un modèle de réactions par un réseau de Petri 
consiste à associer aux espèces chimiques des $places$ et aux réactions des $transitions$. 

\begin{example}
  \label{ex:enz}
 Le modèle réactionnel de la réaction enzymatique de Michaelis-Menten
  \lstinline|A +  E <=> AE => B +  E|
  correspond au réseau de Petri de la Figure~\ref{fig:enz}.
\end{example}

\begin{figure}[htb]
\begin{center}
\begin{tikzpicture}
  \node[place,tokens=3] (A) at (1, 1.5) {};
  \node(placeA)at  (0.4,1.8){A};
  \node[place] (B) at (7, 1.5) {B};
  \node[place,tokens=2] (E) at (2.5, 4) {};
    \node(placeE)at  (1.9,4.3){E};
  \node[place] (AE) at (4, 1.5) {AE};

  \node[transition] (r1) at (2.5, 2.5) {$t_1$}
     edge[pre] (A)
     edge[pre] (E)
     edge[post] (AE);
  \node[transition] (r2) at (2.5, 0.5) {$t_{-1}$}
     edge[pre] (AE)
     edge[post] (A);
  \node[transition] (r3) at (5.5, 1.5) {$t_2$}
     edge[pre] (AE)
     edge[post] (B);

  \draw[round, post] (r2) -- (0, 0.5) -- (0, 4) -- (E);
  \draw[round, post] (r3) -- (5.5, 4) -- (E);
\end{tikzpicture}
\end{center}
\caption{Modèle biochimique de l'Exemple~\ref{ex:enz},
représenté comme un réseau de Petri avec un marquage validant \(t_1\).}
\label{fig:enz}
\end{figure}
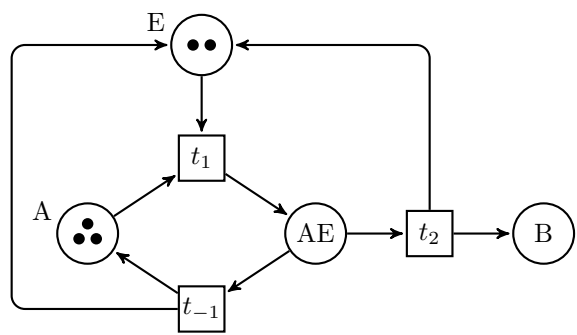
\subsection{Siphons et Pièges}

Soit $PN = (P, T, W)$ un réseau de Petri.

\begin{definition}

Un piège est un ensemble non vide de places $P' \subseteq P$
tels que ses successeurs sont aussi des prédécesseurs: $\Succ{P'} \subseteq \Pred{P'}$.

Un siphon est un ensemble non vide de places $P' \subseteq P$
tel que ses prédécesseurs sont aussi des successeurs. $\Pred{P'} \subseteq \Succ{P'}$.
\end{definition}

Remarquons qu'un siphon dans $PN$ est un piège dans le réseau de Petri dual 
obtenu en inversant la direction de tous les arcs de $PN$.
Notons aussi que puisque les successeurs (resp.\ prédécesseurs) d'une union 
sont l'union des successeurs (resp.\ prédécesseurs), l'union de deux siphons (resp.\ pièges) est un siphon (resp.\ piège).

Les propositions suivantes montrent que les pièges et les siphons donnent 
une caractérisation structurelle de certaines propriétés dynamiques des marquages.

\begin{proposition}\cite{Peterson81book}
Pour tout sous-ensemble de places $P'\subseteq P$,
$P'$ est un piège si et seulement si pour tout marquage
$m \in \N^P$ tel que $m_p \geq 1$ pour certaines places $p \in P'$, et pour tout marquage $m'\in \N^P$ tel que
$m\overset{\sigma}{\rightarrow} m'$
pour une séquence de transitions $\sigma$,
il existe une place $p' \in P'$ telle que $m'_{p'} \geq 1$.
\end{proposition}

\begin{proposition}\cite{Peterson81book}
Pour tout sous-ensemble de places $P'\subseteq P$, $P'$ est un siphon si et seulement si
pour tout marquage $m \in \N^P$ avec $m_p = 0$ pour tout $p \in P'$,
et tout marquage $m'\in \N^P$ tel que 
$m\overset{\sigma}{\rightarrow} m'$
pour une séquence de transitions $\sigma$,
nous avons $m'_{p'} = 0$ pour tout $p' \in P'$.
\end{proposition}

Un exemple d'utilisation des siphons et des pièges a été introduit dans\,\cite{ZS03insilicobio}:
la plante de la pomme de terre produit l'amidon et l'accumule pendant la 
croissance pour ensuite le consommer après la récolte.
Le réseau de Petri correspondant est représenté par la
Figure~\ref{fig:ZS03insilicobio},
où $S_1$ représente le glucose-1-phosphate,
$S_2$ est UDP-glucose et $S_3$ est l'amidon\,\cite{Stryer95book}.
Dans ce modèle, l'une des deux branches, soit la branche produisant l'amidon ($t_3$ et $t_4$), soit la branche le consommant ($t_5$ et $t_6$), est opérationnelle.
$P_1$ et $P_2$ représentent des métabolites externes.

Il est facile d'observer que l'ensemble $\{S_2, S_3\}$ est un piège lorsque $t_3$
et $t_4$ sont opérationnelles et que $\{S_3, S_4\}$ est un siphon lorsque $t_5$ et
$t_6$ sont opérationnelles: une fois un jeton arrive dans $S_3$, aucune transition ne peut être franchie et le jeton y reste 
indépendamment de l'évolution du système.
D'autre part, une fois le dernier jeton est consommé de $S_3$ et $S_4$, 
aucune transition ne pourra générer un nouveau jeton dans ces places qui demeureront vides.

Dans beaucoup de cellules contenant de l'amidon, ce dernier et certains de ces prédécesseurs forment des
pièges, alors que l'amidon et certains successeurs forment des siphons.
En effet, soit la branche produisant l'amidon, soit la branche qui le consomme, est opérationnelle et ceci est réalisé par
une inactivation complète des enzymes appropriées.

\begin{figure}
\begin {center}
      \begin{tikzpicture}[node distance=9mm and 18mm,on grid]
 \node[place] (s1) {$S_1$};       
  \node[place] (p1) [left=of s1, xshift=-2cm]{$P_1$};
        
         \node[place] (p2) [right=of s1, xshift=2cm] {$P_2$};
         \node[place] (s2) [above right=of s1, yshift=1cm] {$S_2$};
			\node[place] (s3) [above=of s1, yshift=3cm] {$S_3$};
         \node[place] (s4) [above left=of s1,  yshift=1cm] {$S_4$};

         \node[transition] (r1) [right=of p1] {$t_1$}
         edge[pre] (p1)
         edge[post, bend left] (p1)   
         edge[post] (s1);   

         \node[transition] (r2) [right=of s1] {$t_2$}
         edge[pre] (s1)
			edge[pre, bend right] (p2)
			edge[post] (p2)
         ;   

         \node[transition] (r3) [above right=of s1,xshift=-1cm] {$t_3$}
         edge[pre] (s1)
         edge[post] (s2);   

         \node[transition] (r4) [above left=of s2,xshift=1cm] {$t_4$}
         edge[pre] (s2)
         edge[post] (s3);   

          \node[transition] (r5) [above right=of s4,xshift=-1cm] {$t_5$}
         edge[pre] (s3)
         edge[post] (s4);  

         \node[transition] (r6) [above left=of s1,xshift=1cm] {$t_6$}
         edge[pre] (s4)
         edge[post] (s1);

      \end{tikzpicture}
   \end{center}
  \caption{Réseau de Petri de\,\protect\cite{ZS03insilicobio} modélisant la croissance de la plante de la pomme de terre.}
   \label{fig:ZS03insilicobio}
\end{figure}
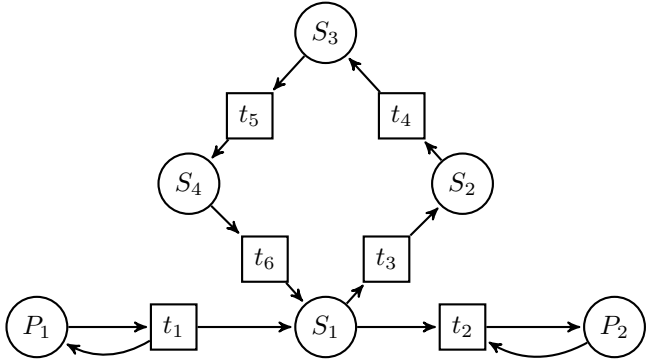

\subsection{Minimalité}

Un siphon (resp.~un piège) est minimal s'il ne contient pas d'autre siphon (resp.~piège).
Malgré la stabilité des siphons et des pièges par union, notons que 
les siphons minimaux ne forment pas un ensemble générique de tous les siphons.

\begin{figure}[htb]
\begin {center}
      \begin{tikzpicture}[node distance=10mm and 2cm,on grid]
         \node[place] (A) {A};
         \node[place] (B) [right=of A] {B};
         \node[place] (C) [right=of B, xshift=1cm] {C};
         \node[place] (D) [right=of C] {D};

         \node[transition] (r1) [above right=of A,xshift=-1cm] {$r_1$}
         edge[pre] (A)
         edge[post] (B);   

         \node[transition] (r2) [below right=of A,xshift=-1cm] {$r_2$}
         edge[pre] (B)
         edge[post] (A);   

         \node[transition] (r3) [right=of B, xshift=-0.5cm] {$r_3$}
         edge[pre] (B)
         edge[post] (C);   

         \node[transition] (r4) [above right=of C,xshift=-1cm] {$r_4$}
         edge[pre] (C)
         edge[post] (D);   

         \node[transition] (r5) [below right=of C,xshift=-1cm] {$r_5$}
         edge[pre] (D)
         edge[post] (C);
      \end{tikzpicture}
   \end{center}
   \caption{Réseau de Petri de l'exemple~\ref{ex:siphon/trap}.}
\label{fig:siphontrap}
\end{figure}
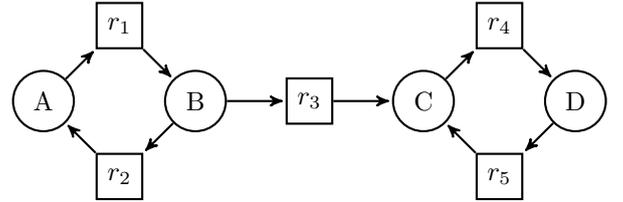 

\begin{example}
\label{ex:siphon/trap}
 Dans le réseau de Petri de la Figure~\ref{fig:siphontrap}, $\{A,B\}$ est
 un siphon minimal: $  {\Pred \{A,B\}}=\{r_1,r_2\} \subset \Succ
{\{A,B\}} =\{r_1,r_2,r_3\}$. $\{C,D\}$ est un piège minimal: $ \Succ {\{C,D\}}
=\{r_4,r_5\} \subset {\Pred \{C,D\}}=\{r_3,r_4,r_5\}$.
\end{example}

Un siphon est générique s'il ne peut pas être représenté sous la forme d'une union d'autres 
siphons \cite{Murata89ieee}. 
Un siphon minimal est un siphon générique, mais un siphon générique n'est pas forcément minimal.
Prenons le cas de l'exemple \ref{ex:siphon/trap},
les siphons sont $\{A,B\}$ et $\{A,B,C,D\}$:
mais seul le premier est minimal, le second ne pourra pas être écrit comme une union de siphons minimaux.

Une motivation pour étudier les siphons minimaux est qu'ils procurent une condition suffisante pour la
non-existence de blocages.
En effet, on a prouvé que dans un réseau de Petri bloqué (i.e.~où aucune transition ne peut être franchie),
toutes les places non-marquées forment un siphon\,\cite{CXra97}.
Ainsi, l'approche basée sur les siphons pour la détection des blocages vérifie si le réseau
contient un siphon propre (un siphon est propre si l'ensemble de ses prédécesseurs est strictement inclus dans l'ensemble de ses successeurs) peut devenir non marqué par une séquence de franchissement.
Un siphon propre ne devient pas non marqué s'il contient un piège initialement marqué.
Si un tel siphon est identifié, le marquage initial est modifié par la séquence de franchissement
et la vérification continue pour les siphons restants jusqu'à ce qu'un blocage soit identifié ou jusqu'à ce que la vérification est terminée.
Considérer uniquement l'ensemble des siphons minimaux est suffisant puisque si un siphon devient non marqué durant l'analyse
alors au moins un siphon minimal doit aussi être non marqué.

D'autres liens avec les propriétés comportementales de vivacité sont établis 
dans\,\cite{HGD08sfm}.

\subsection{Complexité}
Décider si un réseau de Petri contient un siphon ou un piège
et en donner un s'il existe est polynomial\,\cite{CFP03ieee}.
Cependant, le problème de décision de l'existence d'un siphon minimal contenant
une place donnée est NP-difficile\,\cite{TYW96ieice}.
De plus, il peut y avoir un nombre exponentiel de siphons et de pièges dans un réseau de Petri 
comme le montre l'exemple suivant.

\begin{example}
\label{ex:expnb}
Dans le réseau de Petri définie par les équations
\lstinline|A1 + B1 => A2 + B2|,
\lstinline|A2 + B2 => A3 + B3|,
\dots,
\lstinline|An + Bn => A1 + B1|.
et représenté par la Figure \ref{fig:exponentialnb},
il y a $2^n$ siphons minimaux et $2^n$ pièges minimaux,
chacun contenant soit $Ai$ soit $Bi$ mais pas les deux en même temps, pour tout $i$.
\end{example}

\begin{figure}
\begin{center}

{\tt
\begin{tikzpicture}[every node/.style={}]
   \matrix[row sep=3mm, column sep=3mm] {
   \node[place] (A1) {A1}; & &
   \node[place] (A2) {A2}; & &
   \node[place] (A3) {A3}; & &
   \node[place] (An) {An}; & \\
   & \node[transition] (T1) {}; &
   & \node[transition] (T2) {}; &
   & \node (T3) {\dots}; &
   & \node[transition] (Tn) {}; \\
   \node[place] (B1) {B1}; & &
   \node[place] (B2) {B2}; & &
   \node[place] (B3) {B3}; & &
   \node[place] (Bn) {Bn}; & \\
   };

   \draw[pre] (T1) -- (A1);
   \draw[pre] (T1) -- (B1);
   \draw[post] (T1) -- (A2);
   \draw[post] (T1) -- (B2);
   \draw[pre] (T2) -- (A2);
   \draw[pre] (T2) -- (B2);
   \draw[post] (T2) -- (A3);
   \draw[post] (T2) -- (B3);
   \draw[pre] (T3) -- (A3);
   \draw[pre] (T3) -- (B3);
   \draw[post] (T3) -- (An);
   \draw[post] (T3) -- (Bn);
   \draw[pre] (Tn) -- (An);
   \draw[pre] (Tn) -- (Bn);

   \draw[round, post] (Tn) -- ++(0, 2) -| (A1.north);
   \draw[round, post] (Tn) -- ++(0, -2) -| (B1.south);
\end{tikzpicture}
}
\end{center}
\caption{ Réseau de Petri du modèle de l'exemple~\ref{ex:expnb}.}
\label{fig:exponentialnb}
\end{figure}
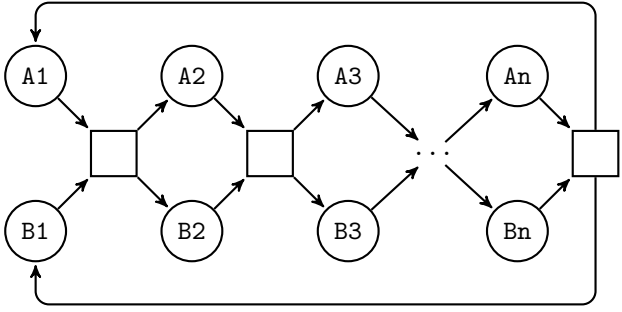

\section{Modèle booléen}
\label{sec:boolean-model}
Dans la littérature, plusieurs algorithmes ont été proposé pour le calcul des siphons et des pièges
minimaux d'un réseau de Petri.
Puisque un siphon dans un réseau de Petri $N$ est un piège dans le réseau dual $N'$, il suffit de traiter les siphons, les pièges sont obtenus par dualité.
Certains algorithmes sont basés sur les inégalités\,\cite{Murata89ieee}, les équations logiques\,\cite{KM86nettheory,MB90dam}, ou 
des approches algébriques\,\cite{L87cnapn}.
Des méthodes plus récentes ont été présentées dans\,\cite{CFP03ieee, CFP05ieee}.

Le problème de recherche d'un siphon peut être naturellement décrit par 
un modèle booléen représentant l'appartenance ou non de chaque place au siphon recherché.
L'énumération de tous les siphons peut être codée comme une procédure de recherche itérative,
similaire à l'approche branch-and-bound grâce à une propriété donnée ci-dessous.

Pour un réseau de Petri contenant $n$ places et $m$ transitions, un siphon $S$ est un ensemble de 
places tel que ses prédécesseurs sont aussi successeurs. $S$ peut être représenté par un vecteur 
$\vec {V}$ de $\{0,1\}^n$ tel que pour tout $i \in \{1,2,..,n\}$,
$V_i = 1$ si et seulement si $p_i \in S$.

La contrainte de siphon peut être formulée comme suit:
\begin{center}
 $  \forall i, V_i=1 \Rightarrow (\forall t \in T, t \in  \Pred
{p_i}
 \Rightarrow  t \in (\cup_{V_j=1} \Succ {\{p_j\})})$.
\end{center}
Cette contrainte est équivalente à:
\begin{center}
 $  \forall i, V_i=1 \Rightarrow \Pred{p_i} \subseteq
\Succ{(\bigcup_{V_j=1} \{p_j\})} $
\end{center}
ce qui peut être écrit sous la forme:
\begin{center}
$  \forall i, V_i=1 \Rightarrow \bigwedge_{t \in
\Pred{p_i}} (\bigvee_{p_j \in \Succ{t}} V_j=1)$.
\end{center}


Finalement, afin d'exclure l'ensemble vide, la contrainte suivante est ajoutée

\begin{center}
$ \bigvee_i V_i=1$.
\end{center}

Énumérer uniquement les siphons minimaux (vis-à-vis de l'inclusion ensembliste)
peut être assuré par la stratégie de recherche et l'ajout de nouvelles contraintes.

La stratégie trouve un siphon minimal vis-à-vis de l'ordre d'inclusion ensembliste et
une nouvelle contrainte est ajoutée à chaque fois qu'un tel siphon est trouvé afin d'interdire 
les siphons le contenant dans la suite de la recherche.

Dans une approche précédente\,\cite{Nabli11cp}, pour la méthode basée sur la programmation avec contraintes, l'ordre d'inclusion était assuré 
par un étiquetage (labeling) sur une variable cardinalité dans un ordre croissant.
Étiqueter directement sur les variables booléennes, par valeur croissante
(d'abord \(0\), puis \(1\)), se révèle être plus efficace, plus facile à appliquer,
et garantit
que les siphons sont trouvés dans l'ordre d'inclusion ensembliste, grâce à la proposition suivante.
\begin{proposition}
Soit un arbre binaire tel que, dans chaque nœud instanciant une variable \(X\)
l'arc gauche poste la contrainte \(X = 0\) et l'arc droit poste la contrainte \(X = 1\),
alors pour toutes les feuilles distinctes \(A\) et \(B\), la feuille \(A\) est à gauche de la feuille
\(B\) seulement si l'ensemble représenté par \(B\) n'est pas inclus dans l'ensemble représenté par \(A\)
(c'est-à-dire, il existe une variable \(X\) telle que \(X_B > X_A\),  
où \(X_A\) et \(X_B\) dénotent les valeurs instanciées à \(X\) dans le trajet 
menant à \(A\) et \(B\) respectivement).
\end{proposition}
\begin{proof}
\(A\) et \(B\) ont au moins un nœud ancêtre en commun instanciant une variable \(X\).
Si la feuille \(A\) est à gauche de la feuille \(B\), l'arc menant à \(A\) est à gauche avec la contrainte \(X = 0\)
et l'arc menant à \(B\) est à droite avec la contrainte \(X = 1\),
par conséquent \(X_B > X_A\).
\end{proof}

À chaque fois qu'un siphon \((S_i)\) est trouvé, la contrainte \(\bigvee_{i \mid S_i = 1} V_i = 0\)
doit être ajoutée au modèle pour garantir que les sur-ensembles de ce siphon ne seront pas énumérés.

\section{Algorithmes}
\label{sec:booleanmodel}
Cette section décrit deux implémentations du  modèle précédent ainsi que deux stratégies de recherche,
une utilisant une procédure SAT itérée et l'autre basée sur la programmation avec contraintes.

\subsection{Algorithme SAT itéré}
Le modèle booléen peut être directement interprété en utilisant un solveur SAT pour 
vérifier l'existence d'un siphon ou un piège.

Nous utilisons \href{http://www.sat4j.org/}{sat4j}, la bibliothèque de satisfiabilité et d'optimisation booléenne pour Java qui fournit une collection de solveurs SAT efficaces.
 Elle inclue une implémentation des spécifications MiniSAT en Java.

Pour le réseau de Petri de la figure~\ref{fig:enz} représentant la réaction enzymatique de l'exemple~\ref{ex:enz}, nous avons le codage suivant: chaque ligne est une liste de variables séparées par des espaces, elle représente une clause.
Une valeur positive signifie que la variable correspondante est sous la forme positive (donc $2$ signifie $V_2$),
et une valeur négative signifie la négation de la variable (donc $-3$ signifie $-V_3$). 
Dans cet exemple, les variables 1, 2, 3 and 4 correspondent respectivement à $E$, $A$, $AE$ et
$B$.
Dans la première itération, le problème est de résoudre les clauses suivantes:\\
-2 3\\
-3 1 2\\
-1 3\\
-1 3\\
-4 3\\
Le problème est satisfiable avec les valeurs: -1, 2, 3, -4 ce qui signifie que $\{A,
AE\}$ est un siphon minimal.
Pour assurer la minimalité, la clause -2 -3 est ajoutée et le programme itère
une autre fois. Le problème est satisfiable avec les valeurs 1, -2, 3, -4, ce qui signifie que $\{E, AE\}$ est aussi un siphon minimal.
Une nouvelle clause est ajoutée précisant que soit $E$ soit $AE$ n'appartient pas au siphon et
aucune affectation des variable ne peut satisfaire le problème.

Ainsi, ce modèle contient 2 siphons minimaux: $\{A, AE\}$ et $\{E, AE\}$.

L'enzyme $E$ est une protéine qui catalyse (i.e., augmente la vitesse de) la 
transformation du substrat $E$ en produit $B$.
L'enzyme est conservée dans une réaction chimique. 

Les résultats obtenus avec cette procédure SAT itérée sont très bons, comme le détaille la section~\ref{sec:eval}.

\subsection{Algorithme PLC(B)}

La recherche de siphons peut aussi être implémentée avec un Programme Logique avec Contraintes
Booléennes (PLC(B)).
Nous utilisons \href{http://www.gprolog.org/}{GNU-Prolog}~\cite{DC01jflp}
pour l'efficacité de ses propagateurs de contraintes booléennes.

La stratégie d'énumération est une variation du \emph{branch-and-bound},
où, à chaque fois qu'un nouveau siphon est trouvé, la recherche doit trouver un non sur-ensemble de ce siphon.

Nous avons essayé deux variantes de branch-and-bound: avec et sans relance.
Dans le branch-and-bound avec relance, il se révèle essentiel de choisir
une méthode de sélection des variables diversifiante.
En effet, les méthodes d'énumération avec un ordre fixe de variables accumulent les échecs
en essayant toujours d'énumérer les mêmes ensembles en premier, qui sont élagués plus tard
par les contraintes de non sur-ensembles.
L'arbre devient ainsi de plus en plus dense à chaque itération puisque plusieurs 
échecs précédents sont explorés à nouveau.
Une sélection aléatoire des variables assure une bonne diversité, tout comme une méthode
qui énumère d'abord sur les variables correspondant aux places figurant dans les siphons déjà trouvés.


\begin{figure*}[htb]
\begin{center}
  \includegraphics[width=\textwidth]{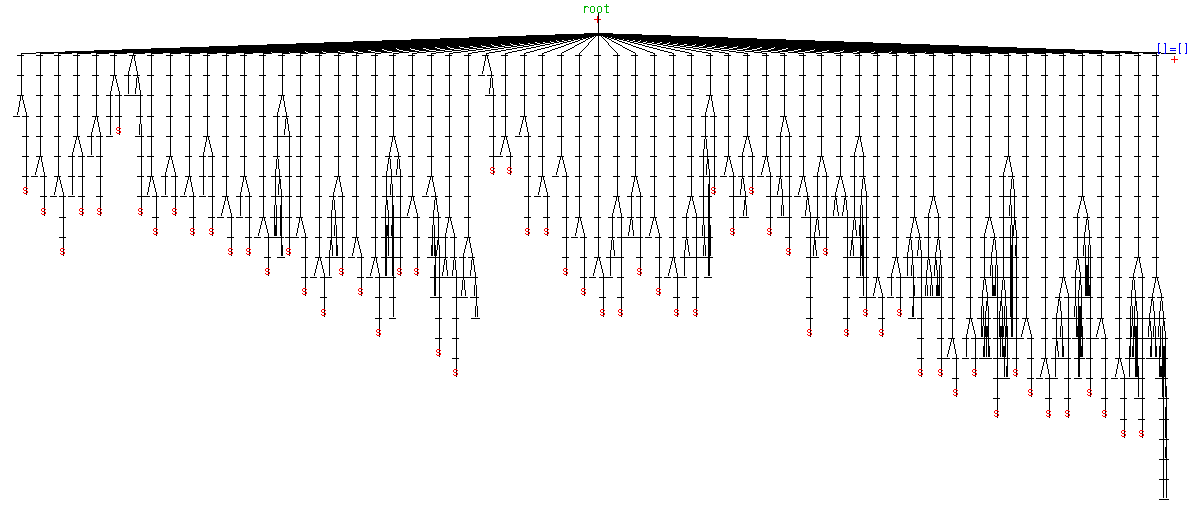}
\end{center}
\caption{Arbre de recherche développé avec la stratégie de backtrack sans relance\label{fig:br}
pour le calcul de tous les siphons minimaux du réseau 1454 de PetriWeb.}
\end{figure*}

Branch-and-bound sans relance donne de meilleurs performances à condition de prendre le soin de 
poster chaque contrainte de non sur-ensembles une seule fois (les reposer toutes à chaque backtrack est inefficace).
Cette stratégie est implémentée de la manière suivante:
à chaque fois qu'un siphon est trouvé, le chemin menant à cette solution est mémorisé, puis 
la recherche est entièrement défaite dans le but d'ajouter au modèle une nouvelle contrainte de non sur-ensemble,
puis le chemin mémorisé est rejoué pour poursuivre la recherche au point où
elle a été arrêtée.
La figure~\ref{fig:br} montre l'arbre de recherche qui est développé,
il est beaucoup plus satisfaisant que celui obtenu avec relance.

Dans une phase de post-traitement, l'ensemble des siphons minimaux peut être filtré afin de ne garder que les siphons 
minimaux qui contiennent un ensemble donné de places, pour résoudre le
problème NP-difficile sus-mentionné.
Notons que poster l'inclusion de l'ensemble sélectionné de places en premier n'assurerait pas que les siphons trouvés sont minimaux vis-à-vis de l'inclusion ensembliste.

\section{Évaluation}
\label{sec:eval}
\subsection{Banc d'essai Petriweb}
\label{sec:benchmarks1}

Notre premier banc d'essai de réseaux de Petri est le dépôt  
\href{http://www.petriweb.org/}{Petriweb}
\cite{GHPW06icatpn}.

Les instances les plus difficiles sont des études de cas dans des processus de raffinement:
\begin{itemize}
\item \emph{"Transit" case study - transit1 process}, identifiant 1454;
\item \emph{"Transit" case study - transit2 process}, identifiant 1479;
\item \emph{"Transit" case study - transit4 process}, identifiant 1516;
\end{itemize}

\subsection{Banc d'essai Biomodels.net}
\label{sec:benchmarks2}

Nous considérons également l'entrepôt \href{http://www.biomodels.net/}{Biomodels.net}
de modèles de réactions biochimiques \cite{NBBCDDLSSSSH06nar} ainsi que d'autres modèles complexes.
Les modèles les plus grands sont comme les suivants:

\begin{itemize}
\item
la carte de Kohn~\cite{Kohn99mbc,CCDFS04tcs} représente une carte de 509 espèces
et 775
interactions moléculaires qui régissent le
cycle cellulaire des mammifères et le mécanisme de réparation de l'ADN;
\item le modèle 175 qui représente les réponses au Ligand du réseau de signalisation de l'ErbB;
\item le modèle 205 qui représente la simulation de la régulation de l'endocytose de l'EGFR
et la signalisation EGFR-ERK; 
\item le modèle 239 qui représente  
un modèle cinétique du réseau de la sécrétion d'insuline stimulée par le glucose
des cellules bêta du pancréas.
\end{itemize}

\subsection{Résultats et Comparaisons}

%
%
%

Dans\,\cite{Nabli11cp}, l'approche PLC a été comparée à un modèle linéaire en nombres entiers (PLNE) \cite{CFP02ieee} sur le dépôt de \href{http://www.biomodels.net/}{Biomodels.net} elle s'est révélée au moins trois fois plus rapide que l'énumération PLNE en utilisant un solveur CPLEX. 
Cet approche PLC implémente le modèle booléen
décrit dans la section ~\ref{sec:booleanmodel}.
Dans cette section, nous comparons deux implémentations de ce modèle booléen à l'algorithme dédié de l'état de l'art que les mêmes auteurs avaient introduit dans\,\cite{CFP05ieee}.

L'algorithme dédié utilise  récursivement une procédure de partitionnement pour
réduire le problème original à plusieurs
sous-problèmes plus simples.
Chaque sous-problème possède des contraintes spécifiques supplémentaires sur
les places par rapport au problème d'origine.
Dans\,\cite{CFP05ieee} la correction, la convergence et la
complexité de l'algorithme  sont montrées.
L'évaluation expérimentale de la performance est également traitée.

Ces algorithmes peuvent être appliqués pour énumérer les siphons minimaux, les
siphons minimaux par rapport à une place ou aussi les siphons minimaux par
rapport à un sous-ensemble donné de places.

Les tables \ref{table:perf} et \ref{table:perfAVG} fournissent les
performances sur les modèles de Petriweb et Biomodels.net.
Dans la table \ref{table:perf}, nous donnons les temps de calcul de tous les siphons minimaux des modèles présentés dans les
sections \ref{sec:benchmarks1} et \ref{sec:benchmarks2}.
Pour chacun de ces exemples, nous fournissons le nombre de siphons minimaux, le nombre de places, 
le nombre de transitions et les temps de calcul en utilisant 
l'algorithme dédié, l'algorithme SAT et le programme PLC.

Toutes les durées sont en ms. Les temps de calcul sont obtenu en utilisant un
PC avec un processeur intel Core 2.20 GHz et 8 Go de mémoire.

Dans la table \ref{table:perfAVG}, nous considérons l'ensemble des modèles mis à disposition
dans les dépôts de Biomodels.net et de Petriweb.
Pour que ces temps totaux aient un sens,
nous avons retiré le modèle numéro 175 de biomodels.net et le réseau numéro 1516 de Petriweb car les temps mesurés sur ces deux modèles sont à la fois très
grands et non représentatifs du cas général : on se reportera
à la table \ref{table:perf} pour les temps obtenus sur ces deux modèles.
Pour chaque base de données, nous donnons le nombre total de modèles,
le nombre minimal et le nombre maximal de siphons trouvés et la moyenne des
nombres des siphons trouvés dans tous les modèles.
Nous donnons également
la taille minimale et la taille maximale des siphons trouvés ainsi que la
moyenne des tailles sur tous les modèles. 
Finalement, nous fournissons le temps de calcul total qui n'est autre que la somme des temps de calcul de chaque modèle.

\begin{table*}[htb]
\begin{center}
\begin{tabular}{|c|c|c|c|c|c|c|}
\hline
 modèle  & \# & \# &\# & algorithme & sat & GNU \\

&siphons & places & transitions & dédié & &Prolog\\
\hline

carte de Kohn &81 & 509 & 775 & 28 & 1 & 221\\
\hline
BIOMD000000175 & 3042 & 118 & 194 & \(\infty\) & 137000 & \(\infty\)\\
\hline
BIOMD000000205 &32 & 194 & 313& 21 & 1 & 34\\
\hline
BIOMD000000239 &64 & 51& 72 & 2980 & 1 & 22\\

\hline
\hline
1454 & 60 & 39 &48& 15  & 1 & 11\\
\hline
1479 & 168 & 58& 78 &47  & 4 & 220\\
\hline
1516 & 1133 & 68	 & 102 & 2072  & 4739 & 163 248\\
\hline
\end{tabular}
\caption{Performance sur les instances les plus difficiles.\label{table:perf}}
\end{center}
\end{table*}

\begin{table*}[htb]
\begin{center}
\begin{tabular}{|c|c|c|c|c|c|c|}
\hline
Base de données & \# & \# siphons & taille des siphons &
\multicolumn{3}{c|}{temps total} \\
\cline{5-7}
& modèles & min-max (avg.) & min-max (avg.) &algorithme &SAT & GNU\\
& & & & dédié &&Prolog\\
\hline

Biomodels.net & 403 & 0-64 (4.21)& 1-413 (3.10) & 19734 & 611 & 231\\
\hline
Petriweb & 79 & 0-168 (6.24) & 1-24 (3.36) & 2757 & 457 & 240 \\

\hline
\end{tabular}
\caption{Performance sur l'ensemble du banc d'essai.\label{table:perfAVG}}
\end{center}
\end{table*}

Sur toutes les instances, notamment les plus larges, le codage SAT est très efficace.
Il convient parfaitement à la taille du réseau ainsi qu'au nombre de siphons minimaux.
L'algorithme dédié est souvent moins efficace avec au moins un ordre de grandeur
de différence, sauf pour une seule instance, le modèle numéro 1516 de biomodels.net, pour laquelle l'algorithme dédié est environ deux fois plus rapide : cette instance est singulière par son grand nombre de siphons minimaux.
Le codage CLP a également une meilleur performance que l'algorithme dédié, 
mais il semble manipuler moins facilement des réseaux de grande taille tels
que la carte de Kohn, et est très lent sur le modèle 1516.




\section{Instances difficiles}
\label{hardins}

MiniSAT dépasse en rapidité l'algorithme spécialisé d'au moins un ordre de grandeur et le temps de calcul est étonnamment court sur nos exemples pratiques.
Même si le modèle est très grand, par exemple la carte de Kohn du contrôle du cycle cellulaire avec 509 espèces et 775 réactions, le 
temps de calcul demeure infime.
Pourtant, cette énumération de tous les siphons minimaux résout le problème de
décision de l'existence d'un siphon minimal contenant un ensemble donné de
places qui a été prouvé NP-difficile dans\,\cite{YW99eice} et la question est: pourquoi le calcul des siphons minimaux
est-il si facile dans les grands réseaux biochimiques ou dans les réseaux de PetriWeb ?

Une façon d'aborder cette question est de considérer la preuve de
NP-difficulté par réduction du problème 3-SAT
et le phénomène de transition de phase dans 3-SAT.
La probabilité qu'un problème 3-SAT aléatoire soit satisfiable suit une transition de phase aigue quand la densité $\alpha$
du nombre de clauses sur le nombre de variables est au voisinage de 4.26 \cite{MSL92AI,CA93AI}, allant de la satisfiabilité vers l'insatisfiabilité avec une probabilité de $1$ quand le nombre de variables tend vers l'infini. 

La réduction de 3-SAT au problème de l'existence d'un siphon minimal traité dans\,\cite{YW99eice} est
obtenue avec des réseaux de Petri dont la structure est illustrée dans la figure~\ref{fig:3-Sat}.
Il est important de noter que dans ce codage, le réseau de Petri a un degré entrant maximum (pour $q_0$) linéaire par rapport au nombre de clauses et un degré maximal sortant (pour $t_0$) linéaire par rapport au nombre de variables.

Sans surprise, cette famille de réseaux de Petri fournit un banc d'essai difficile pour l'énumération des siphons minimaux.
La Table \ref{table:perfgen} contient les résultats expérimentaux de ces réseaux de Petri générés \footnote{Ce banc d'essai est disponible sur \href{http://contraintes.inria.fr/~nabli/indexhardinstances.html}{ce lien}.}.
Nous considérons un time-out de deux secondes, le symbole "-" signifie que le délai de deux secondes n'a pas suffit pour 
énumérer tous les siphons minimaux.
Les résultats de cette section sont obtenus en utilisant un PC avec un processeur intel Core2 Quad 2.8 GHz et 8 Go de mémoire.
Cette table contient les informations qui concernent le codage 3-SAT et le réseau de Petri correspondant: pour chaque 3-SAT, nous
fournissons le nombre de variables booléennes, le nombre de clauses et le ratio $\alpha$, pour le réseau de Petri correspondant nous donnons le nombre de places, le nombre de transitions et la densité (ratio du nombre de transitions sur le nombre de places).
Le temps de calcul de l'énumération de tous les siphons minimaux en fonction de $\alpha$ est représenté dans la figure \ref{fig:3satcurve}; il est visible que le temps de calcul reste exponentiel autour de la valeur critique 4.26 de $\alpha$.

De même, des réseaux de Petri générés aléatoirement avec 
des degrés linéaires par rapport au nombre de sommets (places et transitions) représentent aussi des instances difficiles. 
Par contre, des réseaux de Petri aléatoires ayant des degrés de l'ordre de $10$ comme c'est le cas des modèles biochimiques sont des instances faciles pour le problème d'énumération des siphons minimaux. 
Ainsi, bien que les modèles de réactions biochimiques soient de grande taille, en terme de nombre de places, leurs degré reste borné
à une petite valeur ce qui explique pourquoi le problème d'énumération des siphons minimaux est si facile en pratique.


\begin{figure*}[htb]

\begin {center}
\begin{tikzpicture}
  \node[place] (q0) at (9, 0) {$q_0$};
  \node[place] (s1) at  (5, 3) {$s_1$};
  \node[place] (s1bar) at (5, 2) {$\bar{s_1}$};
  \node[place] (s2) at  (5, 1) {$s_2$};
  \node[place] (s2bar) at (5, 0) {$\bar{s_2}$};
  \node[place] (s3) at  (5, -1) {$s_3$};
  \node[place] (s3bar) at (5, -2) {$\bar{s_3}$};
  \draw [dotted] (5,-2.5)--(5,-3);0
  \node[place] (sn) at  (5, -3.5) {$s_n$};
  \node[place] (snbar) at (5, -4.5) {$\bar{s_n}$};

\node[place] (r1) at  (2, 3) {$r_1$};
  \node[place] (r1bar) at (2, 2) {$\bar{r_1}$};
  \node[place] (r2) at  (2, 1) {$r_2$};
  \node[place] (r2bar) at (2, 0) {$\bar{r_2}$};
  \node[place] (r3) at  (2, -1) {$r_3$};
  \node[place] (r3bar) at (2, -2) {$\bar{r_3}$};
  \draw [dotted] (2,-2.5)--(2,-3);
  \node[place] (rn) at  (2, -3.5) {$r_n$};
  \node[place] (rnbar) at (2, -4.5) {$\bar{r_n}$};

  \node[transition] (u1) at (7, 1.5) {$u_1$}
	edge[post] (q0)     
     edge[pre] (s1)
     edge[pre] (s2)
     edge[pre] (s3);
  \node[transition] (u2) at (7, 0) {$u_{2}$}
  edge[post] (q0)
     edge[pre] (s1bar)
     edge[pre] (s3)
     edge[pre] (sn);
  \node[transition] (u3) at (7, -1.5) {$u_3$}
  edge[post] (q0)
     edge[pre] (s2bar)
      edge[pre] (s3bar)
       edge[pre] (snbar);
 \draw [dotted] (7,-2.3)--(7,-3);
\node[transition] (ualphan) at (7, -3.5) {$u_{\alpha*n}$}
	edge[post](q0)
	edge[pre](s2bar)
	edge[pre](s3bar)
	edge[pre](snbar);
	
 \node[transition] (y1) at (3.5, 3) {$y_1$}
  	 edge[post] (s1)
     edge[pre] (r1)
     edge[pre] (s1bar); 
  \node[transition] (y1bar) at (3.5, 2) {$\bar{y_1}$}
  	 edge[post] (s1bar)
     edge[pre] (r1bar)
     edge[pre] (s1);
\node[transition] (y2) at (3.5, 1) {$y_2$}
  	 edge[post] (s2)
     edge[pre] (r2)
     edge[pre] (s2bar);
\node[transition] (y2bar) at (3.5, 0) {$\bar{y_2}$}
  	 edge[post] (s2bar)
     edge[pre] (r2bar)
     edge[pre] (s2);
 
  \node[transition] (y3) at (3.5, -1) {$y_3$}
  	 edge[post] (s3)
     edge[pre] (r3)
     edge[pre] (s3bar);    
  \node[transition] (y3bar) at (3.5, -2) {$\bar{y_3}$}
  	 edge[post] (s3bar)
     edge[pre] (r3bar)
     edge[pre] (s3);
     \draw [dotted] (3.5,-2.5)--(3.5,-3);
      \node[transition] (yn) at (3.5, -3.5) {$\bar{y_n}$}
  	 edge[post] (sn)
     edge[pre] (rn)
     edge[pre] (snbar);    
  \node[transition] (ynbar) at (3.5, -4.5) {$\bar{y_n}$}
  	 edge[post] (snbar)
     edge[pre] (rnbar)
     edge[pre] (sn);
     
   \node[transition] (t0) at (-1, 0) {$t_0$}
  	 edge[post] (r1)
     edge[post] (r1bar)
    edge[post] (r2)
     edge[post] (r2bar)
      edge[post] (r3)
     edge[post] (r3bar)
      edge[post] (rn)
     edge[post] (rnbar);
\draw[round, post] (q0)  -- (9, -5.3) --(-1, -5.3)-- (t0);

\end{tikzpicture}
   \end{center}
   \caption{Réseaux de Petri de la réduction 3-SAT}
   \label{fig:3-Sat}
\end{figure*}   

\begin{table*}[htb]
\begin{center}
\begin{tabular}{|c|c|c|c|c|c|c|c|c|}
\hline
model &\# & \multicolumn{3}{|c|}{Petri net view} &\multicolumn{3}{|c|}{3-SAT view} & time\\

\cline{3-8}
  & siphons & \# places & \# transitions& density& \# variables & \# clauses& $\alpha$ & (ms) \\
\hline
pn0.xml & 201 & 801& 401& 0,5 &200	&0&	0&	79\\
\hline
pn0.0.xml &	201	&801	&401	&0,5	&200	&0&	0&	79\\
\hline
pn0.2.xml	&-	&801	&441&	0,56&	200&	40&	0,2&	2000\\
\hline
pn0.6.xml&	-	&801&	521	&0,65&	200&	120&	0,6&	2000\\
\hline
pn1.xml&	-	&801&	601&	0,751&	200&	200&	1&	2000\\
\hline
pn2.xml	&-	&801&	801&	1&	200	&400&	2	&2000\\
\hline
pn3.xml	& -	&801	&1001&	1,24	&200&	600	&3 &	2000\\
\hline
pn4.xml	& -&	801	&1201&	1,49	&200&	800	&4&	2000\\
\hline
pn4.2.xml &	-	&801	&1241&	1,54 &	200&	840	&4,2&	2000\\
\hline
pn4.4.xml &	200	&801 &	1281	&1,59	&200&	880	&4,4&	1596\\ 
\hline
pn4.6.xml&	200	&801&	1321	&1,64&	200&	920&	4,6&	1411\\
\hline
pn5.xml&	200&	801	&1401&	1,74	&200	&1000&	5&	370\\
\hline
pn6.xml	&200&	801	&1601&	1,99&	200&	1200	&6&	175\\
\hline
pn7.xml&	200&	801&	1801&	2,24&	200&	1400&	7	&157\\
\hline
pn8.xml	&200	&801&	2001&	2,49&	200&	1600&	8&	157\\
\hline
pn9.xml	&200&	801&	2201&	2,74	&200	&1800&	9	&133\\
\hline
pn10.xml&	200	&801&	2401	&2,99&	200&	2000&	10&	137\\
\hline
\end{tabular}
\caption{Performance sur le banc d'essai généré.\label{table:perfgen}}
\end{center}
\end{table*}

\begin{figure}[htb]
  \includegraphics[width=\columnwidth]{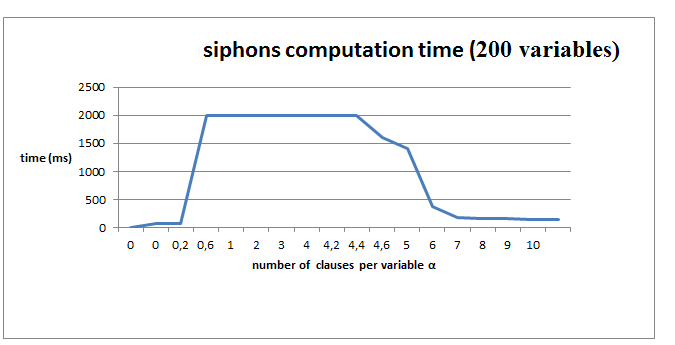}
\caption{\label{fig:3satcurve} le temps de l'énumération de tous les siphons minimaux avec un time-out de 
2 secondes établi une transition de phase au voisinage de la valeur de 4.26 du ratio du nombre de clauses par le nombre de variables.}
\end{figure}

\section{Conclusion}
\label{sec:concl}

Les siphons et les pièges définissent des groupement significatifs de
composants qui exposent un comportement particulier durant l'évolution dynamique 
d'un système biochimique quelles que soient les valeurs des paramètres du modèle.

Nous avons décrit un modèle booléen pour ce problème et nous avons comparé
deux méthodes pour satisfaire ces contraintes de façon à énumérer l'ensemble des
siphons et des pièges minimaux.

Le solveur SAT est le plus efficace sur les réseaux de grande taille.
Le programme PLC(B) est néanmoins plus efficace d'au moins un ordre de grandeur
que l'algorithme dédié de l'état de l'art \cite{CFP05ieee}.

Les deux méthodes sont capables de résoudre tous les problèmes des dépôts de
Petriweb et de \href{http://www.biomodels.net/}{Biomodels.net}
(daté mars 2012) dans un temps court ce qui démontre l'applicabilité de cette
approche.
Notons aussi que le modèle PLC(B) pour calculer les siphons et les pièges
minimaux est une extension du modèle PLC(DF) pour la recherche des P- et 
T-invariants \cite{Soliman08wcb}.

L'efficacité étonnante de ces méthodes appliquées aux modèles biochimiques de biomodels.net a été expliquée en montrant
que le degré des réseaux de Petri est borné dans les modèles biologiques. 

L'idée d'appliquer les solveurs SAT ou la programmation logique par contraintes
aux problèmes classiques de la communauté des réseaux de Petri n'est pas
nouvelle, mais ces approches ont jusque-là été davantage appliquées à la vérification de
modèles. 
Nous sommes persuadés que les problèmes structurels peuvent
également bénéficier du savoir-faire développé dans la communauté de la
modélisation booléenne.
Cela semble être particulièrement intéressant pour la communauté de la biologie
des systèmes, où les progrès technologiques récents nécessitent
de plus en plus d'outils d'analyse puissants et d'expertise pour des problèmes d'optimisation.

\bibliography{contraintes}

\end{document}